\documentclass[12pt]{amsart}
\usepackage{mathrsfs,amsmath,amssymb,amsfonts,amsthm,latexsym,bm}
\usepackage{bbold}

\usepackage[all]{xy}
\usepackage{color}

\theoremstyle{plain}
\newtheorem{theorem}{Theorem}[section]

\theoremstyle{definition}
\newtheorem{remark}[theorem]{Remark}

\newtheorem{definition}[theorem]{Definition}

\newcommand{\abs}[1]{\lvert#1\rvert}
\newcommand{\norm}[1]{\lVert#1\rVert}
\newcommand{\bigabs}[1]{\bigl\lvert#1\bigr\rvert}

\newcommand{\term}[1]{{\textit{\textbf{#1}}}}

\author[N. Gao]{Niushan Gao}
\address{School of Mathematics, Southwest Jiaotong University, Chengdu, Sichuan, 610000, China.}
\email{ngao@home.swjtu.edu.cn}

\author[F. Xanthos]{Foivos Xanthos}
\address{Department of Mathematics, Ryerson University, 350 Victoria St., Toronto, ON, M5B 2K3, Canada.}
\email{foivos@ryerson.ca}

\begin{document}

\title[$w^*$-representations of risk measures]{On the C-property and $w^*$-representations of risk measures}
\date{\today}

\thanks{The authors were in part supported by an NSERC grant.}
\maketitle

\begin{abstract}
We identify a large class of Orlicz spaces $L_\Phi(\mu)$ for which the topology $\sigma(L_\Phi(\mu),L_\Phi(\mu)_n^\sim)$ fails the C-property introduced in \cite{Biagini:10}.
We also establish a variant of the C-property and use it to prove a $w^*$-representation theorem for proper convex increasing functionals, satisfying a suitable version of Delbaen's Fatou property, on Orlicz spaces $L_\Phi(\mu)$ with $\lim_{t\rightarrow\infty}\frac{\Phi(t)}{t}=\infty$. Our results apply, in particular, to risk measures on all Orlicz spaces $L_\Phi(\mathbb{P})$ other than $L_1(\mathbb{P})$.
\end{abstract}

\section{Introduction}

The notion of coherent risk measures was introduced by Artzner et al in \cite{Art:99}. It was later extended to the more general notion of convex and monetary risk measures (see e.g. \cite[Chapter~4]{Follmer:04} and the references therein).
An important topic in the theory of risk measures is to study when the measures under investigation admit certain robust representations, and as risk measures have convexity, a lot of efforts have been devoted to the more general study of representations of proper convex increasing functionals. For example, Delbaen's classical representation theorems (\cite[Theorems~2.3 and 3.2]{Delbaen:02}) on $L_\infty(\mathbb{P})$ can be rephrased as follows (cf.~\cite[Remark~4.17 and Theorem~4.31]{Follmer:04}).

\begin{theorem}\label{del}
\begin{enumerate}
\item\label{del1} Any convex increasing functional $\phi:L_\infty(\mathbb{P})\rightarrow \mathbb{R}$ admits the representation
$\phi(f)=\sup_{\mathbb{Q} \in (L_\infty(\mathbb{P})^*)_+} (\langle \mathbb{Q},f\rangle-\phi^*(\mathbb{Q}))$ for any $f\in L_\infty(\mathbb{P})$,
where $\phi^*(\mathbb{Q})=\sup_{f\in L_\infty(\mathbb{P})}(\langle \mathbb{Q},f\rangle -\phi(f))$ for any $\mathbb{Q}\in (L_\infty(\mathbb{P})^*)_+$.
\item\label{del2} A proper convex increasing functional $\phi:L_\infty(\mathbb{P})\rightarrow(-\infty,\infty]$ admits the representation
$\phi(f)=\sup_{g \in L_1(\mathbb{P})_+} (\langle g,f\rangle-\phi^*(g))$ for any $f\in L_\infty(\mathbb{P})$, where $\phi^*(g)=\sup_{f\in L_\infty(\mathbb{P})}(\langle g,f\rangle -\phi(f))$ for any $g \in L_1(\mathbb{P})_+$,
iff $\phi$ satisfies the \emph{Fatou property}: $\phi(f)\leq \liminf_n \phi(f_n)$ for any $f\in L_\infty(\mathbb{P})$ and any bounded sequence $(f_n)$ in $L_\infty(\mathbb{P})$ such that $f_n\xrightarrow{a.e.~}f$.
\end{enumerate}
\end{theorem}

Such representations have been extensively studied by various authors for function spaces beyond $L_\infty(\mathbb{P})$, as most models in finance and insurance mathematics involve unbounded random variables; see e.g.~\cite{Arai:10,Arai:14,Biagini:10,Biagini:11,Cheridito:09,Delbaen:09,Fri:02,Jouini06,Orihuela:12}.
In particular, Theorem~\ref{del}\eqref{del1} has been fully generalized to proper convex increasing functionals on Banach lattices (\cite[Theorem~4.1]{Cheridito:09}) and on Frechet lattices (\cite[Theorem~1]{Biagini:10}).
The extension of Theorem~\ref{del}\eqref{del2} is subtler. It requires a suitable generalization of Delbaen's Fatou property using the lattice theory terminology. Note that the boundedness of a sequence $(f_n)$ in $L_\infty(\mathbb{P})$ has two equivalent interpretations: norm boundedness, or order boundedness (i.e.~there exists $F\in L_\infty(\mathbb{P})$ such that $\abs{f_n}\leq F$ for all $n\geq 1$). Note also that a sequence in a function space is order bounded and a.e.~convergent iff it is order convergent. In view of these, the authors of \cite{Biagini:10} interpreted the conditions in Delbaen's Fatou property as $f_n\xrightarrow{o}f$ in $L_\infty(\mathbb{P})$, and established the following result which has received significant attention in the mathematical finance literature.

\begin{theorem}\cite[Proposition~1]{Biagini:10}\label{biaa}
Let $\phi:X\rightarrow(-\infty,\infty]$ be a proper convex increasing functional on a Banach lattice $X$. Let $X_n^\sim$ be the order continuous dual of $X$. Suppose that the topology $\sigma(X,X_n^\sim)$ satisfies the $C$-property. Then $\phi$ admits the representation $\phi(x)=\sup_{x^*\in (X_n^\sim)_+}(\langle x^*,x\rangle-\phi^*(x^*))$ for any $x\in X$, where $\phi^*(x^*)=\sup_{x\in X}(\langle x^*,x\rangle-\phi(x))$ for any $x^*\in (X_n^\sim)_+$, iff $\phi$ is $\sigma$-order lower semi-continuous, i.e.~$\phi(x)\leq \liminf \phi(x_n)$ whenever $x_n\xrightarrow{o}x$ in $X$.
\end{theorem}

The C-property introduced in \cite{Biagini:10} can be equivalently rephrased as follows.

\begin{definition}{\cite[Definition~3]{Biagini:10}}\label{CP}
A linear topology $\tau$ on a vector lattice $X$ is said to have the
\term{C-property} if for any convex set $C$ in $X$ and any
$x\in\overline{C}^\tau$, there exists a sequence $(x_n)$ in $C$
such that $x_n\xrightarrow{o}x$.
\end{definition}

It was claimed in \cite[Corollary~4]{Biagini:10} that the topology $\sigma(X,X^\sim_n)$ has the C-property whenever $X$ is an ideal in some $L_1(\mu)$-space, and in particular, when $X$ is an Orlicz space over a finite measure space (\cite[p.~18]{Biagini:10}).
However, as was observed in \cite[Remark~1.5]{Owari:14}, the proof of \cite[Lemma~6]{Biagini:10} has a gap, and thus it is not clear whether $\sigma(X,X^\sim_n)$ has the C-property even when $X=L_\infty(\mathbb{P})$ or other Orlicz spaces. As a consequence, it is not clear whether Theorem~\ref{biaa} extends Theorem~\ref{del}\eqref{del2}.

\bigskip

In this note, we prove that if the conjugate function $\Psi$ of an Orlicz function $\Phi$ satisfies the $\Delta_2$-condition,
then the topology $\sigma(L_\Phi(\mu),L_\Phi(\mu)_n^\sim)$ satisfies the C-property iff $L_\Phi(\mu)$ is reflexive.
It follows that $\sigma(L_\infty(\mathbb{P}),L_\infty(\mathbb{P})_n^\sim)$ satisfies the C-property iff $L_\infty(\mathbb{P})$ is finite-dimensional. As a consequence, Theorem~\ref{biaa} does not extend Theorem~\ref{del}\eqref{del2}.
We also establish a variant of the C-property for Orlicz spaces $L_\Phi(\mu)$ with $\lim_{t\rightarrow\infty}\frac{\Phi(t)}{t}=\infty$ with respect to the $w^*$-topology, and apply this variant to establish a $w^*$-representation theorem for proper convex increasing functionals on such Orlicz spaces that extends Theorem~\ref{del}\eqref{del2}.

It deserves pointing out that our result suggests that one may understand the boundedness condition in Delbaen's Fatou property as norm boundedness.

\subsection*{Notations and Facts}
We refer to \cite{ALIP:06} for all unexplained terminology, notations and standard facts on vector and Banach lattices.
It is well-known that the norm dual, $X^*$, of a Banach lattice $X$ equals its order dual, $X^\sim$ (\cite[Corollary~4.5]{ALIP:06}). The \term{order continuous dual}, $X_n^\sim$, of $X$ is the collection of all linear functionals $x^*\in X^*$ which are order continuous, i.e.~$x^*(x_\alpha)\rightarrow 0$ whenever $x_\alpha\xrightarrow{o} 0$ in $X$. It is well-known that $X \subset (X^*)_n^\sim$. Indeed, for any $x\in X$, if $x^*_\alpha\downarrow 0$ in $X^*$, then $\langle x^*_\alpha,x\rangle=x^*_\alpha(x_+)-x^*_\alpha(x_-)\rightarrow 0$ by \cite[Theorem~1.18]{ALIP:06}, and thus $x\in (X^*)_n^\sim$ by \cite[Theorem~1.56]{ALIP:06}. A Banach lattice $X$ is said to be \term{order continuous} if $x_\alpha\downarrow 0$ in $X$ implies $\norm{x_\alpha}\downarrow0$, or equivalently, if $X^*=X_n^\sim$ (\cite[Theorem~2.4.2]{Meyer-Nieberg:91}), and is called a \term{KB-space} if every norm bounded increasing sequence in $X_+$ is norm convergent, or equivalently, if $X=(X^*)_n^\sim$ (\cite[Theorem~4.60]{ALIP:06}). A KB-space is order continuous (cf.~the paragraph following \cite[Definition~4.58]{ALIP:06}); a dual Banach lattice is KB iff it is order continuous (\cite[Theorem~4.59]{ALIP:06}).

We refer to \cite[Chapter 2]{Edgar:02} for all the terminology and facts on Orlicz spaces used in this note. \term{Throughout this note, $\mu$ (resp.~$\mathbb{P}$) stands for a $\sigma$-finite (resp.~probability) measure over some measurable space $(\Omega,\mathscr{F})$}.
Recall that a function $\Phi:[0,\infty) \rightarrow[0,\infty]$ is called an \term{Orlicz function} if it is left continuous, increasing, convex and non-trivial and $\Phi(0)=0$. 
We say that an Orlicz function $\Phi$ satisfies the \term{$\Delta_2$-condition} at $\infty$ (resp.~at $0$) if there exist $u_0 \in (0,\infty)$ and $k \in \mathbb{R}$ such that $\Phi(2u)<k \Phi(u)$ for all $u \geq u_0$ (resp.~for all $u\leq u_0$). The \term{conjugate}, $\Psi$, of $\Phi$ is also an Orlicz function and is defined by $\Psi(s)=\sup\{ts-\Phi(t) : t \geq 0\}$. The \term{Orlicz space} $L_\Phi(\mu)$ is the space of all a.e.~real-valued measurable functions $f$ (modulo a.e.~equality) such that $\norm{f}_\Phi:=\inf\left\{\lambda>0:\int_\Omega\Phi\left(\frac{|f|}{\lambda}\right)\mathrm{d}\mu\leq 1\right\}<\infty$. This norm $||\cdot||_{\Phi}$ on $L_\Phi(\mu)$ is called the \term{Luxemburg norm} and is equivalent to the \term{Orlicz norm} (cf.~\cite[p.\,60-61]{Edgar:02}). The set $H_\Phi(\mu)$ of all $f \in L_\Phi(\mu)$ such that $\int_\Omega\Phi\left(\frac{|f|}{\lambda}\right)\mathrm{d}\mu<\infty$ for all $\lambda>0$ is called the heart of $L_\Phi(\mu)$. The spaces $L_\Phi(\mu)$ and $H_\Phi(\mu)$ are Banach lattices.

\section{Results}\label{sec2}

In this section, we will consider only dual Orlicz spaces $L_\Phi(\mu)$. Recall that $\Psi$ is finite-valued iff $\lim_{t\rightarrow \infty}\frac{\Phi(t)}{t}=\infty$ (\cite[p.\,35]{Edgar:02}). In this case, it follows from \cite[Theorem~2.2.11]{Edgar:02} that $L_\Phi(\mu)$ with the Orlicz norm is the norm dual of $H_\Psi(\mu)$ with the Luxemburg norm.  Namely,  $L_\Phi(\mu)=H_\Psi(\mu)^*$, where duality is given by integration. Note also that $H_\Psi(\mu)$ is order continuous (\cite[Theorem~2.1.14]{Edgar:02}).

\begin{remark}
Recall from \cite{Biagini:10} that the weak topology on a Banach lattice $X$ satisfies the C-property. Indeed, if $C $ is a convex set in $X$ and $x\in
\overline{C}^w$, then by Mazur's theorem,
$x\in\overline{C}^w=\overline{C}^{\norm{\cdot}}$, so that there exists a sequence $(x_n)$ in $C$ such that
$\norm{x_n-x}\rightarrow0$. Now \cite[Lemma~3.11]{GAOX:14} or \cite[Lemma~4]{Biagini:10}
yields a subsequence $(x_{n_k})$ of $(x_n)$ such that
$x_{n_k}\xrightarrow{o}x$ in $X$.
\end{remark}

\begin{theorem}
Let $\Psi$ be the conjugate of the Orlicz function $\Phi$. Suppose either $\Psi$ satisfies the $\Delta_2$-condition at $0$ and $\infty$ or ``$\mu$ is finite and $\Psi$ satisfies the $\Delta_2$-condition at $\infty$''. Then $\sigma(L_\Phi(\mu),L_\Phi(\mu)_n^\sim) $ satisfies the C-property if and only if $L_\Phi(\mu)$ is reflexive if and only if $L_\Phi(\mu)$ is order continuous. In particular, $\sigma(L_\infty(\mu),L_\infty(\mu)_n^\sim)$ satisfies the $C$-property if and only if $L_\infty(\mu)$ is finite-dimensional.
\end{theorem}

\begin{proof}
Observe that $\Psi$ is finite now (cf.~\cite[last paragraph on p.\,44]{Edgar:02}). Thus we have $L_\Phi(\mu)=H_\Psi(\mu)^*$. We now claim that $H_\Psi(\mu)$ is a KB-space. This, together with \cite[Theorem~4.70]{ALIP:06}, implies that $L_\Phi(\mu)$ is reflexive if and only if it is a KB-space, and being a dual space already, if and only if it is order continuous. For the proof of the claim, observe first that $H_\Psi(\mu)=L_\Psi(\mu)$ by \cite[Theorem~2.1.17]{Edgar:02}. Let $(f_n)$ be a norm bounded increasing sequence in $L_\Psi(\mu)_+$. Then its pointwise limit $f$ belongs to $L_\Psi(\mu)$, by \cite[Theorem~2.1.11(c)]{Edgar:02}. It is clear that $f-f_n\downarrow 0$ in $L_\Psi(\mu)$; thus, $\norm{f_n-f}_{\Psi}\rightarrow 0$, by order continuity of $L_\Psi(\mu)=H_\Psi(\mu)$. This proves the claim.

Now if $L_\Phi(\mu)$ is order continuous, then $L_\Phi(\mu)_n^\sim=L_\Phi(\mu)^*$, and thus the topology $\sigma(L_\Phi(\mu),L_\Phi(\mu)_n^\sim) $ is just the weak topology on $L_\Phi(\mu)$. Therefore, it has the C-property, by the previous remark.
Assume now that $\sigma(L_\Phi(\mu),L_\Phi(\mu)_n^\sim) $ has the C-property. We prove that $L_\Phi(\mu)$ is order continuous. Suppose, otherwise, that $L_\Phi(\mu)$ is not order continuous.
Then $\ell_1$ embeds complementably in $H_\Psi(\mu)$ by
\cite[Theorem~4.69]{ALIP:06}, so that $H_\Psi(\mu)=\ell_1 \oplus Z$ for
some closed subspace $Z$, and $L_\Phi(\mu)=\ell_\infty \oplus Z^*$. Since
$\ell_1$ is a non-quasi-reflexive separable Banach space, we
have, by Ostrovskij's Theorem (\cite[Theorem~2.34]{Hajek}), that
there exists a subspace $W$ in $\ell_\infty$ such that
$\overline{W}^{\sigma(\ell_\infty,\ell_1)}=\ell_\infty$ and for
some $x \in \ell_\infty$, no sequence in $W$ can converge to $x$ in the
$\sigma(\ell_\infty,\ell_1)$-topology. It is straightforward verifications that $\overline{W}^{\sigma(\ell_\infty,\ell_1)}=\overline{W}^{\sigma(L_\Phi(\mu),H_\Psi(\mu))}$. Moreover, since $H_\Psi(\mu)$ is a KB-space, we have that $H_\Psi(\mu)=L_\Phi(\mu)_n^\sim$. Therefore, $$x\in
\overline{W}^{\sigma(\ell_\infty,\ell_1)}=\overline{W}^{\sigma(L_\Phi(\mu),H_\Psi(\mu))}=\overline{W}^{\sigma(L_\Phi(\mu),L_\Phi(\mu)_n^\sim)}.$$ The C-property yields a sequence $(x_n)$ in $W$ such that
$x_n\xrightarrow{o}x$ in $L_\Phi(\mu)$. It follows that $\langle x_n,g\rangle\rightarrow \langle x,g\rangle$ for all $g\in H_\Psi(\mu)$, by $H_\Psi(\mu)=L_\Phi(\mu)_n^\sim$ again. In particular, $\langle x_n,y\rangle\rightarrow \langle x,y\rangle$ for all $y\in \ell_1$, i.e.~ $x_n\xrightarrow{\sigma(\ell_\infty,\ell_1)}x$, contradicting the choice of $x$. This proves that $L_\Phi(\mu)$ is order continuous.

Finally, put $\Phi(t)=0$ for $0\leq t\leq 1$ and $\Phi(t)=\infty$ otherwise. Then $L_\Phi(\mu)=L_\infty(\mu)$, and $\Psi(s)=s$ for all $s\in [0,\infty)$, in particular, $\Psi$ satisfies the $\Delta_2$-condition at $0$ and $\infty$.
Therefore, $\sigma(L_\infty(\mu),L_\infty(\mu)_n^\sim)$ satisfies the
$C$-property if and only if $L_\infty(\mu)$ is reflexive if and only if $L_\infty(\mu)$ is finite-dimensional (cf. \cite[Theorem~5.83]{ALIP:06} and \cite[Theorem~5.85]{ALIP:06} for AM-spaces).
\end{proof}

We now establish a weaker form of the $C$-property for Orlicz spaces $L_\Phi(\mu)$ with $\lim_{t\rightarrow \infty}\frac{\Phi(t)}{t}=\infty$ with respect to the weak-star topology.

\begin{theorem}\label{uo-wstar}
Let $\Phi$ be an Orlicz function such that $\lim_{t\rightarrow \infty}\frac{\Phi(t)}{t}=\infty$. Then for any convex set $C$ of $L_\Phi(\mu)$ and any $f\in\overline{C}^{w^*}$, there exists a sequence $(f_n)$ in $C$ such
that $f_n\xrightarrow{a.e.}f$.
\end{theorem}

\begin{proof}
We first observe that $L_\Phi(\mu)$ contains a positive function $f_0$ which is everywhere non-zero. Indeed, since $\Phi$ is not identically $\infty$, it is easily seen that $\chi_A\in L_\Phi(\mu)$ for any measurable set $A$ of finite measure. By $\sigma$-finiteness of $\mu$, we can decompose $\Omega=\cup_nA_n$ where all $A_n$'s have finite measures. Now take a sequence $(\delta_n)$ of small enough positive real numbers such that $\sum_n\delta_n\chi_{A_n}$ converges in $L_\Phi(\mu)$, then $f_0:=\sum_n\delta_n\chi_{A_n}$ is as desired. This idea, in junction with \cite[Theorem~2.1.14(b)]{Edgar:02}, also yields that $H_\Psi(\mu)$ contains a positive function $g_0$ which is everywhere non-zero.

Being order continuous, $H_\Psi(\mu)$ is an ideal of $L_\Phi(\mu)^*$, by \cite[Theorem~4.9]{ALIP:06}. Thus, the topological dual of $\big(L_\Phi(\mu),\abs{\sigma}(L_\Phi(\mu),H_\Psi(\mu))\big)$ is precisely $H_\Psi(\mu)$, by \cite[Theorem~3.50]{ALIP:06}. Therefore, it follows from Mazur's theorem (cf.~\cite[Theorem~3.13]{ALIP:06}) that $$f\in \overline{C}^{w^*}=\overline{C}^{\sigma(L_\Phi(\mu),H_\Psi(\mu))}=\overline{C}^{\abs{\sigma}(L_\Phi(\mu),H_\Psi(\mu))}.$$
Consequently, there exists a net $(f_\alpha)$ in $C$ such that $f_\alpha\rightarrow f$ in the $\abs{\sigma}(L_\Phi(\mu),H_\Psi(\mu))$-topology, that is, $\langle\abs{f_\alpha-f},g\rangle\rightarrow0$ for any $g\in H_\Psi(\mu)_+$. In particular, $\langle\abs{f_\alpha-f},g_0\rangle\rightarrow0$.
Take $(\alpha_n)$ such that $\langle\abs{f_{\alpha_n}-f},g_0\rangle\leq\frac{1}{2^n}$ for all $n\geq 1$.
Note that $\bigvee_{m=n}^k(\abs{f_{\alpha_m}-f}\wedge f_0)\big\uparrow_k \sup_{m\geq n}(\abs{f_{\alpha_m}-f}\wedge f_0)$. Thus since $g_0\in H_\Psi(\mu)\subset L_\Phi(\mu)_n^\sim$, we have that 
\begin{align*}
\Big\langle\sup_{m\geq n}(\abs{f_{\alpha_m}-f}\wedge f_0),g_0\Big\rangle=&
\lim_k \Big\langle\bigvee_{m=n}^k(\abs{f_{\alpha_m}-f}\wedge f_0),g_0\Big\rangle\leq\lim_k\Big\langle\sum_{m=n}^k(\abs{f_{\alpha_m}-f}\wedge f_0),g_0\Big\rangle\\
\leq& \frac{1}{2^{n-1}}.
\end{align*}
It follows that
$\int_\Omega \big( \inf_{n\geq 1}\sup_{m\geq n}(\abs{f_{\alpha_m}-f}\wedge f_0)\big)g_0\mathrm{d}\mu=\big\langle \inf_{n\geq 1}\sup_{m\geq n}(\abs{f_{\alpha_m}-f}\wedge f_0),g_0\big\rangle=0$,
implying that $\limsup_n(\abs{f_{\alpha_n}-f}\wedge f_0)=\inf_{n\geq 1}\sup_{m\geq n}(\abs{f_{\alpha_m}-f}\wedge f_0)=0$ a.e. Since $f_0>0$ everywhere, it follows immediately that $f_{\alpha_n}\xrightarrow{a.e.}f$.
\end{proof}

We are now ready to present the following $w^*$-representation theorem for proper convex increasing functionals. It extends Theorem~\ref{del}\eqref{del2}; simply recall that if $\Phi=0$ on $[0,1]$ and $\infty$ elsewhere, then $L_\Phi(\mu)=L_\infty(\mu)$.

\begin{theorem}\label{orlicz-rep}
Let $\Phi$ be an Orlicz function such that $\lim_{t\rightarrow \infty}\frac{\Phi(t)}{t}=\infty$. For any proper convex increasing functional $\phi:L_\Phi(\mu)\rightarrow(-\infty,\infty]$, the following are equivalent.
\begin{enumerate}
\item\label{w*-rep1} $\phi$ is $w^*$-lower semi-continuous.
\item\label{w*-rep2} $\phi$ admits the representation $\phi(f)=\sup_{g\in H_\Psi(\mu)_+}\left(\int_\Omega fg\mathrm{d}\mu -\phi^*(g)\right)$, for any $f\in L_\Phi(\mu)$, where $\phi^*(g)=\sup_{f\in L_\Phi(\mu)}\left(\int_\Omega fg\mathrm{d}\mu-\phi(f)\right)$, for each $g\in H_\Psi(\mu)_+$.
\item\label{w*-rep3} $\phi(f)\leq \liminf_n \phi(f_n)$ whenever $\sup_n\norm{f_n}_\Phi<\infty$ and $f_n\xrightarrow{a.e.}f$.
\end{enumerate}
\end{theorem}

\begin{proof}
The proof of \eqref{w*-rep1}$\Leftrightarrow$\eqref{w*-rep2} is standard; we include a proof here for the convenience of the reader.
Applying Fenchel's formula (\cite[Theorem~1.11]{Brezis:11}) to $(L_\Phi(\mu),w^*)$, we have that $\phi$ is $w^*$-lower semi-continuous if and only if $\phi(f)=\sup_{g\in H_\Psi(\mu)}(\langle
f,g\rangle-\phi^*(g))$ for any $f\in L_\Phi(\mu)$, where $\phi^*(g)=\sup_{f\in L_\Phi(\mu)}(\langle f,g\rangle-\phi(f))$ for each $g\in H_\Psi(\mu)$.
Thus, for the equivalence of \eqref{w*-rep1} and \eqref{w*-rep2}, it is sufficient to prove that if $\phi^*(g)<\infty$ then $g\geq 0$ a.e.
Suppose, otherwise, that $\phi^*(g)<\infty$ but $g< 0$ on a set $A$ of positive measure. Without loss of generality, assume that $A$ has finite measure, and put $f=\chi_A$. Then $f\in L_\Phi(\mu)_+$ and $\langle f,g\rangle<0$. Pick $\tilde{f}\in L_\Phi(\mu)$ such that $\phi(\tilde{f})<\infty$.
By definition of $\phi^*$, for any real number $\lambda<0$, we have $\lambda \langle f,g\rangle +\langle \tilde{f},g\rangle=\langle \lambda f+\tilde{f},g\rangle\leq \phi^*(g)+\phi(\lambda f+\tilde{f})\leq \phi^*(g)+\phi(\tilde{f})<\infty $. Letting $\lambda\rightarrow-\infty$, we get a contradiction.

Assume now \eqref{w*-rep1} holds. 
Let $f\in L_\Phi(\mu)$ and $(f_n)$ be a norm bounded sequence in $L_\Phi(\mu)$ such that $f_n\xrightarrow{a.e.}f$. 
We claim that $f_n\xrightarrow{w^*}f$ in $L_\Phi(\mu)$.\footnote{A special case of this claim can be found in \cite[Proposition 6, p.\,148]{Rao:91}.} Indeed, for any $g\in H_\Psi(\mu)$ and any $\varepsilon>0$, by \cite[Theorem~4.18]{ALIP:06}, there exists $f_0\in L_\Phi(\mu)_+$ such that $$\langle\abs{f_n-f},\abs{g}\rangle-\langle\abs{f_n-f}\wedge f_0,\abs{g}\rangle=\langle(\abs{f_n-f}-f_0)^+,\abs{g}\rangle<\varepsilon, \mbox{ for all }n\geq 1;$$here we use the identity $a-a\wedge b=(a-b)^+$. 
Therefore,
$$\bigabs{\langle f_n-f,g\rangle} \leq \langle \abs{f_n-f},\abs{g}\rangle<\langle\abs{f_n-f}\wedge f_0,\abs{g}\rangle+\varepsilon.$$
Now put $\widetilde{f}_n:=\sup_{m\geq n}(\abs{f_m-f}\wedge f_0)$. Since $0\leq \widetilde{f}_n\leq f_0$, we have that $\widetilde{f}_n\in L_\Phi(\mu)$; since $f_n\xrightarrow{a.e.} f$, we have that $\widetilde{f}_n\downarrow0$. 
Therefore, it follows from
$\abs{f_n-f}\wedge f_0\leq \widetilde{f}_n$ that $\abs{f_n-f}\wedge f_0\xrightarrow{o}0$ in $L_\Phi(\mu)$. 
Now since $\abs{g}\in H_\Psi(\mu)\subset L_\Phi(\mu)_n^\sim$, it follows that $\langle\abs{f_n-f}\wedge f_0,\abs{g}\rangle\rightarrow 0$, and therefore, $\limsup_n\bigabs{\langle f_n-f,g\rangle} \leq\varepsilon$. 
By arbitrariness of $\varepsilon$, we obtain that $\lim_n\langle f_n-f,g\rangle=0$. This proves the claim. Now by $w^*$-lower semi-continuity of $\phi$, we have $\phi(f)\leq \liminf_n\phi(f_n)$. 
This proves \eqref{w*-rep1}$\Rightarrow$\eqref{w*-rep3}.

Suppose now \eqref{w*-rep3} holds. 
For the implication \eqref{w*-rep3}$\Rightarrow$\eqref{w*-rep1}, we need to prove that the set $A_\lambda:=\{f\in L_\Phi(\mu): \phi(f)\leq \lambda\}$ is $w^*$-closed for each $\lambda\in\mathbb{R}$. 
Observe first that each $A_\lambda$ is convex. Hence, by \cite[Theorem~4.44]{Fabian:01}, in order to prove it is $w^*$-closed, it is sufficient to prove that $A_\lambda\cap mB$ is $w^*$-closed for each $m\geq 1$, where $B$ is the closed unit ball of $L_\Phi(\mu)$. 
Indeed, for any $f\in \overline{A_\lambda\cap mB}^{w*}$, Theorem~\ref{uo-wstar} yields a sequence $(f_n)$ in $A_\lambda\cap mB$ such that $f_n\xrightarrow{a.e.}f$. 
It follows from assumption that $\phi(f)\leq \liminf_n \phi(f_n)\leq \lambda$, so that $f\in A_\lambda$. 
By the standard fact that closed balls are $w^*$-closed, we also have that $f\in mB$, and therefore, $f\in A_\lambda\cap mB$. 
It follows that $A_\lambda\cap mB$ is $w^*$-closed.
\end{proof}

\begin{remark}
The condition $\lim_{t\rightarrow \infty}\frac{\Phi(t)}{t}=\infty$ is very mild. For example, when considering Orlicz spaces on a probability space $(\Omega,\mathscr{F},\mathbb{P})$, which is a frequently used framework in Mathematical Finance, this condition is satisfied whenever $L_\Phi(\mathbb{P})\neq L_1(\mathbb{P})$ (cf.~\cite[Proposition~2.2.6(1) and (2.1.21) on p.\,48]{Edgar:02}). We also mention that the space $L_1[0,1]$ is excluded for a good reason: it is never a dual space.
\end{remark}

Finally, we remark that most of our results hold in the general framework of Banach lattices; here we need to replace a.e.~convergence with the notion of unbounded order convergence, which is recently developed in \cite{GAO:14,GAOX:14,GAOTX:15}.

\bigskip

\textbf{Acknowledgements.} The authors would like to thank Dr.~Patrick Beissner for bringing the reference \cite{Biagini:10} into their attention. The first author would also like to thank Ryerson University for the hospitality received during his visit.

\end{document}